\documentclass[%
reprint, %linenumbers,
 amsmath,amssymb,
 aps, %physrev,
prl
]{revtex4-2}

\usepackage{libpapers}
\usepackage{libgraphics}
\usepackage{booktabs}
\usetikzlibrary{quantikz2}
\usepackage{float}
\usepackage{dcolumn}

\def\I{\mathbb{I}}

\usepackage{pifont} % \ding for check / cross marks
\newcommand{\yes}{\textcolor{green!55!black}{\ding{51}}}
\newcommand{\no}{\textcolor{red!70!black}{\ding{55}}}
\newcommand{\partialmark}{\textcolor{orange!85!black}{\textbullet}}

\newcommand{\gr}[1]{\textcolor{black}{#1}}%{\textcolor{green!50!black}{#1}}

\begin{document}

%\preprint{APS/123-QED}

%\title{\textbf{Quantum eigenvalue transformation for general matrices via quantum signal processing}}%
\title{Quantum Eigenvalue Transformations for Arbitrary Matrices}

\author{Xabier Guti\'errez}
\affiliation{Department of Physical Chemistry, University of the Basque Country UPV/EHU, Apartado 644, 48080 Bilbao, Spain}
\affiliation{EHU Quantum Center, University of the Basque Country UPV/EHU, Apartado 644, 48080 Bilbao, Spain}

\author{Lorenzo Laneve}
\affiliation{Faculty of Informatics --- Università della Svizzera Italiana, 6900 Lugano, Switzerland}

\author{Mikel Sanz}
\affiliation{Department of Physical Chemistry, University of the Basque Country UPV/EHU, Apartado 644, 48080 Bilbao, Spain}
\affiliation{EHU Quantum Center, University of the Basque Country UPV/EHU, Apartado 644, 48080 Bilbao, Spain}
\affiliation{Basque Center for Applied Mathematics (BCAM), Alameda de Mazarredo 14, 48009 Bilbao, Spain}


\date{\today}% It is always \today, today,

\begin{abstract}
Quantum Signal Processing (QSP) and Quantum Singular Value Transformation (QSVT) provide an efficient framework for implementing polynomials of block-encoded matrices, and thus offer a systematic approach to quantum algorithm design. However, despite a number of recent advances, important limitations remain. In particular, QSP can only transform unitary matrices, by applying a polynomial to their eigenvalues, while QSVT is a singular-value transformation and thus one can only obtain the polynomial of Hermitian matrices. As a consequence, these techniques do not directly apply to an arbitrary non-Hermitian matrix that is not diagonalizable. In this work, we propose a simple yet powerful method to extend these ideas to arbitrary square matrices by acting on their eigenvalues. To this end, we introduce the notion of an $n$-regular block encoding, namely, a block encoding whose $k$-th power reproduces the $k$-th power of the encoded matrix for every $0 \leq k \leq n$. We show that applying QSP to any unitary with this property is equivalent to applying a polynomial of degree at most $n$ to the block-encoded matrix, independently of its internal structure. Moreover, we provide a simple construction that transforms any block encoding into an $n$-regular one using only $\mathcal{O}(\log n)$ ancillary qubits and $\mathcal{O}(n\log (n/\epsilon))$ operations. Finally, we show that this construction induces the desired transformation on the eigenvalues associated with the Jordan normal form of the matrix\gr{; for non-diagonalizable matrices this yields the full matrix function $P(A)$, whose action on each Jordan block involves the derivatives of $P$, and not merely the map $\lambda \mapsto P(\lambda)$}.
\end{abstract}
% \bb comments cannot appear inside the abstract environment in revtex without risking layout issues, so:
% [Note on abstract, in blue below \maketitle]

\maketitle

% quantum algorithms and development techniques
\textit{Introduction\textemdash} Quantum algorithms have the potential to deliver substantial speedups over their classical counterparts, yet until very recently their design had no truly systematic framework and instead relied largely on case-by-case ingenuity. In this context, the method of \emph{quantum signal processing} (QSP)~\cite{lowOptimalHamiltonianSimulation2017,martynGrandUnificationQuantum2021} has emerged as a unifying paradigm, linking algorithm design with techniques from numerical analysis and, in particular, with the implementation of polynomials of matrices encoded through unitary transformations. This perspective enabled a more abstract and systematic approach to algorithm design, and it has already been applied in a variety of settings, including Hamiltonian simulation~\cite{lowOptimalHamiltonianSimulation2017,lowHamiltonianSimulationQubitization2019}, quantum search~\cite{martynGrandUnificationQuantum2021}, quantum linear systems~\cite{childsQuantumAlgorithmSystems2017}, and quantum state preparation~\cite{linNearoptimalGroundState2020,dongGroundStatePreparationEnergy2022,mcardleQuantumStatePreparation2022,laneveRobustBlackboxQuantumstate2023}.

In QSP, one has access to a \emph{signal} $z$---typically a complex number of unit modulus or, depending on the convention, a real number in $[-1,1]$---encoded in a single-qubit unitary known as the \emph{signal} operator. By interleaving applications of this operator with a sequence of operations that is independent of $z$---the \emph{processing} operators---one can construct essentially any polynomial transformation $P(z)=\sum_{k=0}^n a_k z^k$ of the signal, which may, for instance, appear in the upper-left entry of the resulting unitary~\cite{lowMethodologyResonantEquiangular2016,gilyenQuantumSingularValue2019a,motlaghGeneralizedQuantumSignal2024}. Moreover, efficient and numerically stable methods are known for constructing a single-qubit protocol realizing a prescribed polynomial $P$~\cite{alexisQuantumSignalProcessing2024,alexisInfiniteQuantumSignal2026,laneveGeneralizedQuantumSignal2025,niInverseNonlinearFast2025}. By leveraging the spectral theorem, this construction can be lifted so that the scalar signal $z$ is replaced by a unitary operator $U$~\cite{motlaghGeneralizedQuantumSignal2024}. The resulting circuit then implements the non-unitary transformation $P(U)=\sum_{k=0}^n a_k U^k$, thereby applying $P$ directly to the eigenvalues of $U$ in a \emph{black-box} fashion, that is, without the need to explicitly construct or estimate the corresponding eigenstates. Remarkably, this requires only $n+1$ single-qubit operations and $n$ controlled uses of $U$.

On the other hand, when a matrix $A$, not necessarily square, appears in the upper-left block of a unitary $U_A$ that we can implement---a so-called \emph{block-encoding}---one can similarly apply a QSP polynomial to its singular values by a closely related construction, known as \emph{quantum singular value transformation} (QSVT)~\cite{gilyenQuantumSingularValue2019a,sunderhaufGeneralizedQuantumSingular2023}. This method has received particular attention because of its practical usefulness. It is important, however, to note that a polynomial of a square matrix does not, in general, coincide with the polynomial transformation of its singular values.

Despite their rapid development, these techniques still suffer from important limitations. In particular, the existing frameworks can perform either an eigenvalue transformation of a unitary matrix, or a singular value transformation of a block-encoded matrix $A$, but they do not directly provide an eigenvalue transformation of a block-encoded matrix, except in the Hermitian case, where singular values and eigenvalues coincide up to sign. Consequently, for a general non-Hermitian matrix, neither method can be used directly to implement its polynomial transformation.

Several approaches beyond the QSP framework have been proposed to partially bridge this gap, especially the construction of Chebyshev and Faber history states~\cite{lowQuantumEigenvalueProcessing2024}, and the method of \emph{linear combination of Hamiltonian simulations} (LCHS)~\cite{anLinearCombinationHamiltonian2023}, which has been shown to be useful for the simulation of non-unitary dynamics~\cite{huangFourierTransformbasedLinear2025,lowOptimalQuantumSimulation2025}. In particular, the latter approach has been used to block-encode eigenvalue transformations by expressing them as inverse Laplace transforms~\cite{anLaplaceTransformBased2026}. \gr{Recently, a Weyl calculus approach to LCHS has been proposed, extending the class of admissible matrices \cite{niQuantumEigenvalueTransformation2026}.}

However, \gr{these} methods rely on the \emph{linear combination of unitaries} (LCU) technique~\cite{childsHamiltonianSimulationUsing2012}, which predates QSP and exhibits several limitations. First, the complexity of implementing a target transformation $f(z)$ depends strongly on its specific form and typically requires a non-trivial synthesis of the corresponding quantum circuit. Moreover, a linear combination of $N$ terms generally requires an $N$-dimensional register, so the qubit overhead can grow rapidly, especially in the aforementioned approaches, where LCU is used to approximate complicated integrals to a prescribed precision. Another drawback of LCU-based methods is that the resulting block-encoded matrices can suffer from severe subnormalization---that is, instead of obtaining a block-encoding of $A$, one obtains a block-encoding of $A/\alpha$ for some potentially very large $\alpha$. This leads to small success probabilities and, consequently, to further overhead in the overall complexity.

In this Article, we propose a method to carry out transformations of eigenvalues, understood as solutions to the characteristic polynomials of a matrix, of an arbitrary square matrix, even non-diagonalizable ones. In order to do this, we introduce the concept of $n$-\emph{regular} block-encoding as a block-encoding whose $k$-th power reproduces the $k$-th power of the encoded matrix for every $0 \leq k \leq n$, and prove that this kind of block-encoding is always possible by adding only $\bigO(\log n)$ ancilla qubits and \gr{$\bigO(n \log(n/\epsilon) )$ extra two-qubit gates}. We then show that a standard application of QSP yields the eigenvalue transformation when this property holds. This allows us to obtain matrix transformations by remaining in this established framework.

This technique extends the QSP formalism to several applications, such as simulation of general inhomogeneous differential equations and non-unitary dynamics, which arise in a broad range of settings, including dissipative dynamics treated via Carleman linearization~\cite{liuEfficientQuantumAlgorithm2021,kroviImprovedQuantumAlgorithms2023}, partial differential equations~\cite{childsHighprecisionQuantumAlgorithms2021}, and open quantum systems~\cite{weimerSimulationMethodsOpen2021}, among others. \gr{In particular, since the cost of our construction for a matrix $A$ is independent of the eigenvector conditioning $\kappa(S)$, where $A = S J S^{-1}$ is the Jordan decomposition of the matrix, it applies without penalty to defective generators (i.e., with high $\kappa(S)$) at exceptional points, a regime of intense current interest in non-Hermitian photonics, PT-symmetric and open quantum systems that is out of reach of previous techniques.}

\textit{Block-encodings and (generalized) quantum signal processing \textemdash}
Here we summarize the main definitions and results that we are going to use in the rest of the work. Let us start by clarifying the notation: we will use $\D$ to denote the complex unit disk $\D = \qty{ z \in \C : \abs{z} < 1 }$, while $\T = \partial \D$ will be the unit circle. Additionally, if not said otherwise, $\norm{\cdot}$ refers to the operator norm. Let us now provide the definition of a block encoding:

% Definition of block-encoding
\begin{definition}
    A unitary $U$ is said to ($a$, $\epsilon$)-\emph{block-encode} a matrix $A$ if
    \begin{align}
        \norm{ (\bra{0}^{\otimes a} \otimes I) U (\ket{0}^{\otimes a} \otimes I) - A } \le \epsilon,
    \end{align}
with $a$ the number of ancillary qubits needed for the block-encoding. In other words, the unitary is of the form
\begin{align}
    U =
    \begin{pmatrix}
        \Tilde{A} & \cdot \\
        \cdot & \cdot
    \end{pmatrix}
\end{align}
for some $\Tilde{A}$ satisfying $\norm{\Tilde{A} - A} \leq \epsilon$.
\end{definition}
Notice that any unitary $(0,0)$-block-encodes itself. Moreover, any block-encoded matrix $A$ must necessarily satisfy $\norm{A} \le 1$ in order for the unitarity condition to hold for $U$.

Within the block-encoding framework, we aim to identify a set of operations that can be used to compose block-encodings and thereby obtain more \emph{useful} matrices from those that are easier to construct. A powerful tool for this purpose is \emph{generalized quantum signal processing} (GQSP). Throughout, $\Tilde{w}$ denotes the signal operator, defined by $\Tilde{w} = \diag(1,z)$, which encodes a complex number $z \in \mathbb{T}$.

\begin{theorem}[GQSP theorem, adapted from~\cite{motlaghGeneralizedQuantumSignal2024}]
    Let $P(z)$ be a polynomial in $z$ of degree $n$ satisfying $\abs{P(z)} \le 1$ on $\mathbb{T}$. There exists a sequence of single-qubit operations $R_0, R_1, \ldots, R_n \in SU(2)$ such that
    \begin{align}
        \bra{0} R_0 \Tilde{w} R_1 \Tilde{w} \cdots \Tilde{w} R_n \ket{0} = P(z)
    \end{align}
\end{theorem}
Moreover, it is always possible, without loss of generality, to choose $R_1, \ldots, R_n$ to be a product of only two Pauli rotations (e.g., $X$ and $Z$), whose angles are now easily computed with numerical algorithms~\cite{alexisQuantumSignalProcessing2024,alexisInfiniteQuantumSignal2026,laneveGeneralizedQuantumSignal2025,niInverseNonlinearFast2025}.

Although previous works have only considered GQSP on the unit circle $\mathbb{T}$, any analytic function $f(z)$ satisfying $\abs{f(z)} \le 1$ on $\mathbb{T}$ automatically satisfies the same bound on $\mathbb{D}$, by virtue of the maximum modulus principle. This shows that, using existing techniques, we are already able to approximate any function $f(z)$ that is analytic and bounded by $1$ on the unit \emph{disk}.

A particularly simple construction enabled by GQSP is the transformation of the eigenvalues of a unitary. Specifically, if $U$ is a $d \times d$ unitary, then its controlled version has the block-diagonal form
\[
C(U) = \diag(I_d, U),
\]
which closely resembles $\Tilde{w}$. In particular, by applying the $R_k$ to the control qubit and replacing $\Tilde{w}$ with $C(U)$, the resulting circuit would simply block-encode $P(U)$, that is, the matrix obtained by mapping each eigenvalue of $U$ through the transformation $z \mapsto P(z)$.

%\section{Regular block-encodings and regularization by `dumping branches'}
\textit{Regular block-encodings and regularization by `dumping branches' \textemdash} Current techniques based on quantum signal processing allow one to transform the eigenvalues of a unitary (as in the construction above) or the singular values of general block-encoded matrices (using QSVT), but no construction based on GQSP was previously known to successfully carry out \emph{eigenvalue} transformations of generic block-encoded square matrices (except in the case of Hermitian matrices, where eigenvalue and singular value decompositions coincide and QSVT can be applied).

Unfortunately, directly performing an eigenvalue transformation $U \mapsto P(U)$ as we did above does not yield a block-encoding of $P(A)$, because, in simple terms, $U^k$ does not generally block-encode $A^k$. Let us formally define this property:

% Definition of 'n-regular' block-encoding
\begin{definition}
    An ($a$, $\epsilon$)-block-encoding $U$ of $A$ is called $n$-regular if $U^k$ ($a$, $k\epsilon$)-block-encodes $A^k$ for every $0 \le k \le n$.
\end{definition}
The bound on the error implied by the definition simply comes from the bound $\norm{A^k - B^k} \le k \norm{A - B}$. These
$n$-regular block encodings show the following crucial property:
\begin{lemma}
    \label{thm:eigenvalue-transformation-robustness}
    Let $P(z)$ be any polynomial of degree $n$ satisfying $\abs{P(z)} \le 1$ on $\mathbb{D}$. If $U$ is a $n$-regular ($a$, $\epsilon$)-block-encoding of $A$, then $P(U)$ is a ($a$, $\sqrt{n^3/3} \epsilon$)-block-encoding of $P(A)$.
\end{lemma}
The proof of this lemma is given in the \gr{Supplementary Material section A}, while a worked example on a 2-regular block-encoding is found in \gr{section B}. This regularity notion is not particularly difficult to achieve, in the sense that any block encoding can be easily transformed into an $n$-regular block-encoding via the simple construction shown in the following theorem:
\begin{theorem}
    \label{thm:block-encoding-regularization}
    Let $U_A$ be an $(a, \epsilon)$-block-encoding of $A$, and suppose $n = 2^b$. Let $Q_n$ denote the quantum incrementer, defined by
\begin{align}
    Q_n \ket{x} = \ket{(x + 1) \bmod n}.
\end{align}
Then, the following circuit yields an $n$-regular $(b + a, \epsilon)$-block-encoding of $A$:
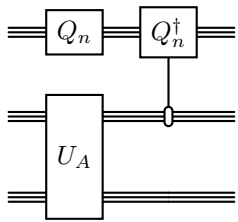
    \begin{center}
        \begin{quantikz}[wire types={b, b, b}, classical gap=0.06cm]
            & \gate{Q_n} & \gate{Q_n^\dag} & \\
            & \gate[2]{U_A} & \octrl{-1} & \\
            & & &
        \end{quantikz}
    \end{center}
\end{theorem}
In other words, the circuit first applies $A$ and then increments the newly introduced $b$-qubit register by~$1$ whenever the ancillary qubits are not in the all‑zero state. The underlying idea is the following: when we apply a block-encoding, we effectively implement the mapping
\begin{align}
    \ket{0}^{\otimes a} \ket{\psi}
    \mapsto
    \ket{0}^{\otimes a} \otimes A \ket{\psi}
    + \sum_{j \neq 0} \ket{j} \otimes \ket{\psi^{\perp}_j},
\end{align}
for some garbage states $\ket{\psi^{\perp}_j}$. This implies that, within the ``successful'' subspace (i.e., the subspace on which $A$ is applied), the associated unitary $U$ leaves $\ket{0}^{\otimes a}$ invariant, whereas all other terms map the ancillary register to non‑zero configurations. Consequently, if $U$ is applied a second time, some of these non‑zero configurations may be mapped back to $\ket{0}^{\otimes a}$, thereby contaminating the successful branch and resulting in the top‑left block encoding $A^2$ plus an additional garbage term. With this construction, however, the incrementer causes the new $b$-qubit register to end up in a non‑zero configuration; thus, even if subsequent calls to the block‑encoding map the original ancillas back to $\ket{0}^{\otimes a}$, the overall main branch associated with $\ket{0}^{\otimes b+a}$ remains clean. This remains true until $Q_n$ is applied $n$ times, after which some branch (in particular, the one that has failed $n$ times) is mapped from $\ket{1}^{\otimes b}$ back to $\ket{0}^{\otimes b}$. This is still acceptable, since the original ancillas are not in the state $\ket{0}^{\otimes a}$; in other words, this yields an $n$-regular block‑encoding, but not an $(n+1)$-regular one, because a branch of the form $\ket{0}^{\otimes b} \ket{k}^{\otimes a}$ can contaminate the main branch. A formal proof is provided in the \gr{Supplementary Material section C}.

Note that, within this approach, standard techniques for multiplying block‑encodings~\cite{fangTimemarchingBasedQuantum2023,dalzellQuantumAlgorithmsSurvey2025} can be directly applied. \gr{The incrementer requires only $\bigO(b) = \bigO\left(\log n\right)$ ancillary qubits and can be diagonalised by an approximate quantum Fourier transform~\cite{nielsenQuantumComputationQuantum2010}, one at the beginning and another one at the end of the circuit, with a cost of $\bigO\left(b\log b\right) = \Tilde{\bigO}\left(\log n \right)$ two-qubit gates. Taking into account each of the $n$ controlled rotations, from which $\log(n/\epsilon)$ are retained, the overall two-qubit gates would be $\mathcal{O}\left(n\log(n/\epsilon)\right)$.}

\textit{Quantum eigenvalue transformation of arbitrary matrices \textemdash} We can thus apply the construction from the previous section to obtain an eigenvalue transformation of~$A$.
\begin{theorem}[Quantum eigenvalue transformation]\label{thm:qet}
    Let $P(z)$ be a polynomial of degree~$n$ and let $A$ be an arbitrary square matrix. Then a block-encoding of $P(A)$ can be obtained using $\bigO(\log n)$ ancillary qubits, $\bigO(n \log n)$ elementary operations, and $n$ calls to a block‑encoding of~$A$.
\end{theorem}

\begin{figure*}
    \begin{center}
        \begin{quantikz}[wire types={q, b, b, b}, classical gap=0.06cm]
            & \gate{R_0} & \ctrl{2} & \ctrl{1} & \gate{R_1} & \ctrl{2} & \ctrl{1} & \ \cdots\  & \ctrl{2} & \ctrl{1} & \gate{R_n} & \\
            && \gate{Q_n} & \gate{Q_n^\dag} & & \gate{Q_n} & \gate{Q_n^\dag} & \ \cdots\  & \gate{Q_n} & \gate{Q_n^\dag} & & \\
            && \gate[2]{U_A} & \octrl{-1} & & \gate[2]{U_A} & \octrl{-1} & \ \cdots\ & \gate[2]{U_A} & \octrl{-1} & & \\
            && & & & & & \ \cdots\  & & &&
        \end{quantikz}
    \end{center}
    \caption{Circuit for quantum eigenvalue transformation of general matrices based on (G)QSP. If $R_0, \ldots, R_n$ are the processing operators realizing some polynomial $P(z)$, then the above circuit implements a block-encoding of $P(A)$, where $A$ is block-encoded by $U_A$.}
    \label{fig:quantum-eigenvalue-transformation-regularized}
\end{figure*}

A full circuit implementing a polynomial transformation of a given matrix~$A$ is depicted in Figure~\ref{fig:quantum-eigenvalue-transformation-regularized}. Notice that, although we use the term \emph{eigenvalue transformation}, we do not even need to assume $A$ is diagonalizable, since we do not rely on the spectral mapping theorem, unlike previous constructions~\cite{lowHamiltonianSimulationQubitization2019, gilyenQuantumSingularValue2019a}. The polynomial transformation achieved by this construction can be better understood via the Jordan decomposition of the matrix, $A = S J S^{-1}$, where $J = \bigoplus_i J_i$. This decomposition remains valid even for non‑diagonalizable matrices. Since any polynomial of $A$ can be written as $P(A) = S P(J) S^{-1}$, with $P(J) = \mathrm{diag}\bigl(P(J_i)\bigr)$, the transformation amounts to a polynomial transformation of the Jordan blocks. In the case when $A$ is diagonalizable, each $J_i$ reduces to the eigenvalue $\lambda_i$, and the transformation coincides with the polynomial transformation of the eigenvalues. \gr{On the other hand, in the case of a non-trivial Jordan block, $P(J_i)$ is the upper-triangular Toeplitz matrix carrying $P^{(k)}(\lambda_i)/k!$ on its $k$-th superdiagonal (see Supplementary Material section C); throughout this Article, ``eigenvalue transformation'' is to be understood in this Jordan sense, i.e., as the implementation of the full matrix function $P(A)$, and not merely of the map $\lambda_i \mapsto P(\lambda_i)$.}

\begin{table*}[t]
\centering
\renewcommand{\arraystretch}{1.3}
\setlength{\tabcolsep}{3pt}
\footnotesize
\resizebox{\textwidth}{!}{%
\begin{tabular}{|c|c|c|c|c|c|c|c|}
\toprule
\textbf{Method} &
\textbf{Spectral object} &
\textbf{Matrix class} &
\textbf{Defective} &
\textbf{Queries to $U_A$} &
\textbf{Extra ancillas} &
\textbf{Success amplitude} &
\textbf{Function class} \\
\midrule
QSVT \cite{gilyenQuantumSingularValue2019a} &
singular values &
general (block-encoded) &
\no &
$n$ &
$\bigO(1)$ &
$\bigO(1)$ (set by $\norm{p}_\D$) &
real definite-parity \\

QET-U \cite{motlaghGeneralizedQuantumSignal2024, dongGroundStatePreparationEnergy2022} &
eigenvalues &
unitary $U$ only &
\no &
$n$ &
$1$ &
$\bigO(1)$ &
any $\abs{P}\le1$ on $\T$ \\

Faber history states~\cite{lowQuantumEigenvalueProcessing2024} &
eigenvalues &
non-normal, spectrum in region $\Omega$ &
\partialmark &
$\bigO(n)$ &
$\bigO(\log n)$ (history reg.) &
Faber-capacity of $\Omega$; bounds $\propto\kappa(S)$ &
Faber series on $\Omega$ \\

Lap-LCHS  \cite{anLaplaceTransformBased2026}&
eigenvalues &
$L=\tfrac{A+A^\dagger}{2}\succeq0$ (stable) &
\partialmark &
$\widetilde{\bigO}(\alpha \norm{g}T\log^{1+o(1)}\tfrac1\varepsilon)$&
$\bigO(\log K+\log T)$, $K{=}\bigO(\log^{1+1/\beta}\tfrac1\varepsilon)$  &
$\propto 1/\norm{f}_1\norm{g}_1$; state-prep optimal &
inv.\ Laplace $g$ via LCU \\

Weyl-LCHS  \cite{niQuantumEigenvalueTransformation2026}&
eigenvalues  &
arbitrary square &
\yes &
 $\mathcal{O}(\frac{B}{\rho}\log \frac{B}{\rho \epsilon})$,&
$\bigO(\log \log 1/\epsilon)$&
$\propto 1/\norm{f}_1$ &
any $f$ analytic around $W(A)$ + margin $\rho$ \\

\textbf{This work} (GQSP $+$ $n$-regular BE) &
eigenvalues / \textbf{full Jordan} &
\textbf{arbitrary square} &
\yes &
$n$ &
$\bigO(\log n)$ &
$\bigO(1)$ if $\abs{f}\!\le\!\text{const on }\D$%;
&
any $f$ bounded analytic on $\D$ \\
\bottomrule
\end{tabular}%
}
\caption{Comparison of spectral-processing methods. ``Defective'' indicates
whether a non-diagonalizable input is admissible with cost \emph{independent of
the eigenvector conditioning} $\kappa(S)$: \yes\ yes; \partialmark\ formally
representable but with hypotheses/bounds that degrade as $\kappa(S)\!\to\!\infty$
(Low \& Su) or that are violated for indefinite Hermitian part / unstable
spectra (Lap-LCHS); \no\ not applicable. Here $f$ is the polynomial function, $n$ is the polynomial degree,
$\alpha$ the block-encoding normalization, $K$ the
number of LCHS quadrature nodes and  $T$ a truncation parameter determined by the decay of the inverse Laplace $g$, $B \geq \max_{\norm{z}\leq 1+\rho}\{\abs{f(z)},\abs{f'(z)},\abs{f''(z)}\}$, $W(A)$ the numerical range of $A$ and $\rho$ is the margin between $W(A)$ and the nearest point where $f$ ceases to be analytical. The success amplitude is the inverse
subnormalization (probability $=$ amplitude$^2$). }
\label{tab:comparison}
\end{table*}

\textit{Examples \textemdash} 
As a first example, we consider the matrix exponential $e^A$. We simply implement the function $f(z) = e^z/3$ via GQSP, which is analytic and satisfies $\abs{f(z)} \le e/3 < 1$ on the unit disk $\mathbb{D}$. The Taylor coefficients $c_k = z^k/k!$ of the exponential decay factorially on $\mathbb{D}$, so that
\begin{align}
    \abs{ e^z - \sum_{k = 0}^N \frac{z^k}{k!} } \le \epsilon
\end{align}
holds on $\mathbb{D}$ already for $N = \mathcal{O}\left(\frac{\log 1/\epsilon}{\log \log 1/\epsilon}\right)$. \gr{Note that the subnormalization is the \emph{constant} $1/3$, independent of the degree and of the precision, so the success amplitude is $\Theta(1)$; in LCU-based implementations the subnormalization is instead set by the $1$-norm of the coefficients of the linear combination.}

In the short‑time simulation regime (i.e., implementing $e^{tA}$ for $t = \bigO(1)$), this method requires fewer copies of $A$ than the LCHS method for simulating time‑independent homogeneous linear differential equations presented in~\cite{lowOptimalQuantumSimulation2025}, and it can be used to address unstable evolutions as well. \gr{The number of ancillary qubits is also significantly lower, see Table \ref{tab:comparison}.}

%\XG{Maybe it is simpler to mention the damped oscillator as a specific example of the matrix exponential}\MS{Agree! See my green edit below introduces it as ``a physically relevant instance of the previous example'', which keeps both while making the link explicit and saving a transition sentence.}
\gr{As a second example, providing a physically relevant instance of the matrix exponential, consider the damped oscillator $\ddot{\boldsymbol{ x}} + C\dot{ \boldsymbol{x}} + K \boldsymbol{x} = 0$,
written in state-space form $\dot{\mathbf v} = A\,\mathbf v$ where, for a one-dimensional system with $C = 2\gamma$ and $K = \omega_0^2$,
\begin{align}
A(\gamma,\omega_0) =
\begin{pmatrix} 0 & 1 \\ -\omega_0^2 & -2\gamma \end{pmatrix}.
\end{align}
The solution is given by $e^{A t}$. At \emph{critical damping} $\gamma=\omega_0$ the two eigenvalues collide at
$\lambda=-\gamma$ and $A$ becomes a single Jordan block: it is
\emph{defective}, $A+A^\dagger \nsucceq 0$ and the eigenvector matrix $S$ satisfies $\kappa(S)\to\infty$.
This is the prototypical \emph{exceptional point} (EP) of a non-Hermitian
generator, ubiquitous in open and PT-symmetric systems, photonic lattices,
and dissipative many-body models. Since our construction never references the eigenvector basis, it implements $e^{tA}$ in short-time regime at the exceptional point with exactly the same cost $\mathcal{O}\!\left(\frac{\log 1/\epsilon}{\log\log 1/\epsilon}\right)$ as in the diagonalizable case, whereas LCHS is inapplicable ($A + A^\dagger \nsucceq 0$) and the Chebyshev/Faber guarantees degrade with $\kappa(S)\to\infty$. Besides having no condition on the transformed matrix, for short-time exponentiation the method presented here has lower query complexity, extra two-qubit gates and extra ancillas compared with the rest of the methods (see Table D.1 in the Supplementary Material for a direct comparison). }

\gr{Finally, consider the wave equation $\ddot{ \boldsymbol{u}} + A \boldsymbol{u} = 0$, whose solution can be written as $u(t) = \cos(t\sqrt{A})u(0) + A^{-1/2}\sin(t\sqrt{A})\dot u(0)$. These functions are holomorphic in the complex unit disk and do not have a limit for $t \to \infty$: thus, there is no locally integrable inverse Laplace transform. These functions have recently been implemented via contour-based matrix decomposition \cite{wangQuantumSimulationNonHermitian2026a}, with a complexity that scales with the condition number and $\sqrt{1/\epsilon}$. We can use Theorem \ref{thm:qet} to directly implement these functions with complexity $ \mathcal{O}\left(\frac{\log 1/\epsilon}{\log \log 1/\epsilon}\right)$ for short times (see Supplementary Material section E).}

\gr{An LCHS-type extension that lifts the stability requirement has recently been proposed \cite{niQuantumEigenvalueTransformation2026}; however, its complexity $\mathcal{O}(S\log(S/\epsilon))$ is governed by $S \geq \max_{z\in\overline{\D}}\{\abs{f},\abs{f'},\abs{f''}\}$, and thus in the short-time regime the complexity is $\bigO(\log1/\epsilon)$, so the direct implementation via Theorem~\ref{thm:qet} remains favourable. } 

In general, we do not have a simple closed‑form expression for the Taylor series of an arbitrary function we wish to implement, so it is not straightforward to bound its convergence rate directly. For convenience, we state the following well‑known result from complex analysis (with a proof provided in the Supplementary Material section F), which turns out to be particularly useful in the application of the technique presented here.

\begin{theorem}[Taylor series convergence]%
    \label{thm:taylor-series-convergence}
    Let $f(z)$ be a function analytic and bounded in $\mathbb{D}$, and suppose $f$ admits an analytic continuation to the disk $\{ z \in \mathbb{C} : \abs{z} < R \}$ for some radius $R = 1 + r > 1$, on which $\abs{f(z)} \le M$. Then the Taylor expansion of $f$ centered at $z = 0$,
    \[
        f(z) = \sum_{k = 0}^\infty a_k z^k,
    \]
    satisfies the error estimate
    \begin{align}
        \abs{ \sum_{k = 0}^N a_k z^k - f(z) } \le \epsilon
    \end{align}
    on $\mathbb{D}$ whenever $N = \Theta(\frac{1}{r} \log \frac{M}{r\epsilon})$.
\end{theorem}
By Theorem~\ref{thm:taylor-series-convergence}, this method ensures $\bigO(\log M/\epsilon)$ copies of $A$ are sufficient to block-encode $f(A)$ up to precision $\epsilon$, whenever $f$ extends analytically across $\T$.

%\section{Conclusion}
\textit{Conclusions \textemdash}
In this Article we propose a method to apply a polynomial to the eigenvalues (or, more generally, to the Jordan blocks) of a block‑encoded square matrix~$A$. We have shown that, if the $k$-th power $U^k$ block‑encodes $A^k$, a property we call \emph{regularity}, then $P(A)$ can be obtained by simply applying $P$ as a GQSP polynomial to $U$, and that any block‑encoding can be turned into an $n$-regular one by equipping it with a quantum incrementer.

\gr{Recently proposed methods cannot be generally applied to any arbitrary square matrix (see Table \ref{tab:comparison} for a comparison). Lap-LCHS can only be applied to stable matrices, with positive semidefinite Hermitian part \footnote{Lap-LCHS cannot in general be applied to nilpotent matrices. As a simple example, consider the Jordan block $A= \begin{pmatrix}    0 & 1 \\ 0 & 0 \end{pmatrix}$, such that all eigenvalues are $0$ but $A+A^\dagger = X \nsucceq 0$}. Besides, it relies on the existence of the inverse Laplace transform for the function class, which might not generally exist (see the example of the wave equation). Chebyshev/Faber history states have a complexity and estimation guarantees that scale with $\kappa(S)$, which
diverges as the matrix approaches defectiveness. Weyl-LCHS is able to handle unstable matrices, but its complexity scales with $\abs{f(z)},\abs{f'(z)},\abs{f''(z)}$, making it inefficient for rapidly varying functions, such as the mentioned  wave-equation solution or step functions.}

\gr{Moreover, the aforementioned methods suffer from the high gate and ancilla cost of LCU, in contrast with the comparatively fewer requirements of the construction in Theorem \ref{thm:block-encoding-regularization}, of $\mathcal{O}(\log n)$ ancilla qubits and $\mathcal{O}(n\log(n/\epsilon))$ two-qubit gates, where $n$ is the maximum degree of the polynomial. Note that the ancilla count for LCHS in Table \ref{tab:comparison} refers only to the implementation of $e^{-At}$; the discretised Laplace transformation increases this cost drastically. We remark that GQSP itself requires one singly-controlled application of $U_A$ per degree---a constant-factor overhead over the uncontrolled block encoding---whereas the select oracles underlying LCU-based methods require multiply-controlled applications of $U_A$ or of its powers; the comparison in Table~\ref{tab:comparison} therefore also holds at the level of controlled queries.} 

We also note that unitaries not only block‑encode themselves, but are also $\infty$-regular, which provides an alternative viewpoint on why applying GQSP to unitaries effectively yields an eigenvalue transformation. At the same time, if we consider a block‑encoding of a matrix $A$ that is close to a unitary, then the block‑encoding will be approximately regular, since the failed branches multiply by blocks of small operator norm.

A limitation of this method lies in the rigidity of the domain: GQSP requires the construction of functions that are analytic on the \emph{entire} unit disk $\mathbb{D}$, even if we can assume that the eigenvalues of $A$ lie in a proper subset of it (for example, the stability region for the simulation of non‑unitary dynamics).  \gr{We can circumvent this problem via a conformal map $\calM(z)$ from the region in which the eigenvalues lie to the unit circle: we first construct $\calM^{-1}(A)$, and apply the function $f(\calM(z))$ to it, with $\calM(z)$ chosen such that the composite function is easier to handle. We provide an example for the exponential in the Supplementary Material section G\gr{, with a total complexity $\mathcal{O}(t^2 + \log 1/\epsilon)$ in the long-time regime.} \iffalse The extra cost is modest: for the Cayley-type map used there, $\calM_\rho^{-1}(A)$ requires inverting $A + I$, whose condition number is $\mathcal{O}(1)$ on the relevant domain, so QSVT implements it with $\mathcal{O}(\log(1/\epsilon))$ additional queries. \fi}
 
\begin{acknowledgments}
\textit{Acknowledgments \textemdash} The authors would like to thank Dong An for important clarifications on~\cite{jiangContourintegralBasedQuantum2026} and to acknowledge insightful discussion with Mario Berta. LL is supported by the Swiss National Science Foundation (SNSF), project No.\ 200020-214808.  XG and MS thank support from the Basque Government BasQ initiative under the Q-STREAM project. They also acknowledge support from OpenSuperQ+100 (Grant No. 101113946) of the EU Flagship on Quantum Technologies, from Project Grant No. PID2024-156808NB-I00 and Spanish Ram\'on y Cajal Grant No. RYC-2020-030503-I funded by MI- CIU/AEI/10.13039/501100011033 and by “ERDF A way of making Europe” and “ERDF Invest in your Future”, and from the Spanish Ministry for Digital Transformation and of Civil Service of the Spanish Government through the QUANTUM ENIA project call Quantum Spain, and by the EU through the Recovery, Transformation and Resilience Plan–Next Generation EU within the framework of the Digital Spain 2026 Agenda.
\end{acknowledgments}

\bibliography{refs}

%\tableofcontents
\pagebreak
\widetext

\section{Supplementary Material}
\appendix

\section{A. Robustness of eigenvalue transformation (proof of Lemma 4)}

\begin{proof}
    For a given polynomial of degree $n$ $P(z) = \sum_{k=0}^n \alpha_k z^k$, if $U$ is block-encoding $\Tilde{A}$ exactly, then by linearity and the definition of $n$-regular encoding it block-encodes $P(\Tilde{A})$ exactly. The error in the block-encoding is bounded by
    \begin{align}
        & \norm{P(\Tilde{A}) - P(A)} \le \sum_{k = 0}^n \abs{\alpha_k} \cdot \norm{\Tilde{A}^k - A^k} \leq \sum_{k=0}^n k \abs{\alpha_k} \epsilon \ .
    \end{align}
    We use Parseval's identity and obtain that $\sum_{k=0}^n \abs{\alpha_k}^2 = (2\pi)^{-1}\int_0^{2\pi}\abs{P(e^{i\theta})}^2 d\theta \leq 1$. From the Cauchy-Schwarz inequality we obtain
    \begin{align}
        \sum_{k=0}^n k\abs{\alpha_k} & \leq \sqrt{\sum_{k=0}^n k^2}\sqrt{\sum_{k=0}^n \abs{\alpha_k}^2} \leq \sqrt{\frac{n(n+1)(2n+1)}{6}}\sim \sqrt{n^3/3} \ .
    \end{align}
    which provides the claimed bound on the error.
\end{proof}

\section{B. Polynomial Transformation of Jordan blocks}
% For A = SJS^{-1}, P(A) = SP(J)S^{-1}, where P(J) in terms of P'(lambda), P''(lambda),...
Any square matrix $A$ can be written in Jordan normal form as $A = SJS^{-1}$ for some invertible matrix $S$, where $J = \bigoplus_i J_i$ is the diagonal of its Jordan blocks $J_i$ of the form:
\begin{align}
    J_i = \begin{pmatrix}
        \lambda_i & 1 & & \\
        & \lambda_i & 1 & \\
        & & \lambda_i & & \\
        & & & \ddots & 1 \\
        & & & & \lambda_i
    \end{pmatrix},
\end{align}
that is, the entries of its diagonal contain the eigenvalue $\lambda_i$, it has $1$s in its superdiagonal and $0$s everywhere else. A matrix being diagonalizable means that each Jordan block is $1 \times 1$.

Applying a polynomial to the matrix $A$ corresponds to a polynomial transformation of its Jordan blocks, as $P(A) = S P(J) S^{-1}$, where $P(J) = \diag(P(J_i))$. The polynomial transformation of a single Jordan block $J_i$ ends up being an upper triangular Toeplitz matrix containing the derivatives $P^{(k)}(\lambda)/k!$ on the $k$-th upper diagonal:
\begin{align}\label{polynomial_Jordan}
    P(J_i) = \begin{pmatrix}
        P(\lambda_i) & P'(\lambda_i) & \frac{P''(\lambda_i)}{2!} & \cdots & \frac{P^{(d)}(\lambda_i)}{d!} \\
        & P(\lambda_i) & P'(\lambda_i) & \cdots & \frac{P^{(d-1)}(\lambda_i)}{(d-1)!} \\
        && P(\lambda_i) && \frac{P^{(d-2)}(\lambda_i)}{(d-2)!} \\
        &&& \ddots & \vdots\\
        && & & P(\lambda_i)
    \end{pmatrix}.
\end{align}
This gives also a natural way to block-encode the Toeplitz matrix of the Taylor coefficients of a given function, in a similar fashion to the Chebyshev and Faber history states of~\cite{lowQuantumEigenvalueProcessing2024}.

As an explicit example of $2$-regular block encoding, consider an arbitrary matrix $A$, block-encoded by
\begin{align}
    U_A =
    \begin{pmatrix}
        A & B\\ C &D
    \end{pmatrix}
\end{align}
and a second-order polynomial $P(z) = (1 + z^2)/2$. In order to implement $P(A)$, we need a $2$-regular block encoding of $A$. We can easily check that generally
\begin{align}
    U_A^2 = \begin{pmatrix}
        A^2 + BC & \cdot \\ \cdot & \cdot
    \end{pmatrix}\neq \begin{pmatrix}
    A^2 & \cdot \\ \cdot & \cdot
\end{pmatrix}.
\end{align}
We apply Theorem 5 with $n = 2$, (here the incrementer $Q_2$ coincides with the Pauli $X$) and obtain the $2$-regular block encoding
\begin{align}
    U_{A,2} = \begin{pmatrix}
        A & B & 0 & 0\\
        0 & 0 & C & D \\
        0 & 0 & A & B\\
        C & D & 0 & 0
    \end{pmatrix},
\end{align}
such that
\begin{align}
    U_{A,2}^2 = \begin{pmatrix}
        A^2 & AB & BC & BD \\
   DC  & D^2 & CA & CB\\
    BC  & BD & A^2 & AB\\
    CA  & CB & DC & D^2
    \end{pmatrix}
\end{align}
is a block encoding of $A^2$. We can directly check that $U_{A,2}^3$ is not a block encoding of $A^3$. Applying (G)QSP we obtain a block encoding of
\begin{align}
P(U_{A,2}) & = (I + U_{A,2}^2)/2  = \begin{pmatrix}
    (I + A^2)/2 & \cdot \\
    \cdot & \cdot
\end{pmatrix} = \begin{pmatrix}
    P(A) & \cdot \\ \cdot & \cdot
\end{pmatrix}.
\end{align}

\section{C. Regularization (proof of Theorem 5)}

\begin{proof}
    We name the three registers in the circuit (from top to bottom) \textbf{C}ounter ancillas, \textbf{O}riginal ancillas and \textbf{S}tate space. In other words, the block-encoded matrix $A$ will act on $S$.

    Denote by $Q$ the operator consisting of the application of the incrementer $Q_n$ and the controlled application of $Q_n^\dagger$, whose action is thus
    \begin{align}
        Q \ket{i}_C \otimes \ket{j}_O =
        \begin{cases}
            \ket{i}_C \otimes \ket{j}_O & j = 0\\
            \ket{i+1}_C \otimes \ket{j}_O & j \neq 0 \ .
        \end{cases}
    \end{align}
    Let's assume for now, for simplicity, that the register $C$ has infinitely many qubits, so that the $i \rightarrow i+1$ operation is not in $\Z_n$ but over ordinary integers. The claim would follow by the fact that the construction modulo $n$ acts like the infinite one for the first $n$ applications of the new construction.

    We now show by induction that $U^k$ block-encodes $A^k$, the case $k = 1$ being trivial. Of course, if we start from the state $\ket{0}_C \ket{0}_O \ket{\psi}_S$, an application of $U^{k-1}$ can generate only terms with $\ket{j}_C$ for $j < k$. Thus, by applying the induction hypothesis, the state after $k-1$ applications can be written as
    \begin{align}
        U^{k-1} \ket{0}_C \ket{0}_O \ket{\psi}_S & = \ket{0}_C \ket{0}_O \otimes A^{k-1} \ket{\psi}_S  + \sum_{0 < j < k} \sum_{l} \ket{j}_C \ket{l}_O \ket*{\phi_{j, l}}_S
    \end{align}
    for some garbage states $\ket*{\phi_{j, l}}_S$. Suppose that now we apply the $k$-th copy of $U$: we compute the action on the two terms separately
    \begin{align}
        U \ket{0}_C \ket{0}_O \otimes A^{k-1} \ket{\psi}_S & = \ket{0}_C \ket{0}_O \otimes A^k \ket{\psi}_S + \ket{1}_C \sum_{l \neq 0} \ket{l}_O \ket*{\beta_l}_S
    \end{align}
    for some garbage states $\ket*{\beta_l}_S$. For the second term ($0 < j < k$), we have
    \begin{align}
        U \ket{j}_C \ket{l}_O \otimes \ket*{\phi_{j, l}}_S & = \ket{j}_C \ket{0}_O \otimes \ket{\gamma_j}_S + \ket{j+1}_C \sum_{l \neq 0} \ket{l}_O \otimes \ket*{\phi'_{j+1, l}}_S
    \end{align}
    for some garbage states $\ket{\gamma_j}_S, \ket*{\phi'_{j+1,l}}_S$. Notice that these terms can never return to the state $\ket{0}_C \ket{0}_O$, as $j > 0$, and $C$ can only be increased by $Q$. This proves the claim, as $\bra{0}_{CO} U^k \ket{0}_{CO} = A^k$.

    If $Q$ is modulo $n$ this is still true up to $U^n$ where some branch makes the transition $\ket{11 \cdots 1}_C \mapsto \ket{00 \cdots 0}_C$. This last $n$-th application still block-encodes $A^n$, as the transition to zero on $C$ means that $O$ is non-zero.
\end{proof}

\section{\gr{D. Comparison with other methods in the short-time regime}}

\begin{table*}[t]
\centering
\renewcommand{\arraystretch}{1.3}
\setlength{\tabcolsep}{3pt}
\footnotesize
\resizebox{\textwidth}{!}{%
\begin{tabular}{|c|c|c|c|c|c|}
\toprule
\textbf{Method} &
\textbf{Condition on $A$} &
\textbf{Queries to $U_A$} &
\textbf{Extra gates} &
\textbf{Extra ancillas} &
\textbf{Success amplitude} \\
\midrule
QSVT~\cite{gilyenQuantumSingularValue2019a}, GQSP~\cite{berryDoublingEfficiencyHamiltonian2024} &
skew-Hermitian &
$\bigO(\frac{\log \frac{1}{\epsilon}}{\log \log \frac{1}{\epsilon}})$ &
$\bigO(\frac{\log \frac{1}{\epsilon}}{\log \log \frac{1}{\epsilon}})$ &
$\bigO(1)$ &
$\Omega(1)$ \\

Faber history states~\cite{lowQuantumEigenvalueProcessing2024} &
non-normal, spectrum in region $\Omega$ &
$\bigO(\kappa(S) \log^2 \frac{1}{\epsilon})$ &
$\bigO(\poly \log \frac{1}{\epsilon})$ &
$\bigO(\log \log \frac{1}{\epsilon})$ &
$\Omega(1)$ \\

LCHS~\cite{lowOptimalQuantumSimulation2025,niQuantumEigenvalueTransformation2026} &
$\frac{A + A^\dag}{2} \succeq 0$ &
$\bigO(\log \frac{1}{\epsilon}) $&
$\bigO(\log^{5/2} \frac{1}{\epsilon} \log \log \frac{1}{\epsilon})$ &
$\bigO(\log \log \frac{1}{\epsilon})$ &
$\Omega(1)$ \\

\textbf{This work} (GQSP $+$ $n$-regular BE) &
arbitrary square &
$\bigO(\frac{\log \frac{1}{\epsilon}}{\log \log \frac{1}{\epsilon}})$ &
$\bigO(\log^2 \frac{1}{\epsilon})$ &
$\bigO(\log \log \frac{1}{\epsilon})$ &
$\Omega(1)$
\\
\bottomrule
\end{tabular}%
}
\caption{Comparison of various methods to block-encode $e^{-At}$ for $t = \bigO(1)$.}
\label{tab:comparison-exponential}
\end{table*}

Here we show how our method achieves better complexity (with respect to the desired error) in every metric when we consider the short-time regime, i.e., $t = \bigO(1)$. As shown in Table~\ref{tab:comparison-exponential}, our method achieves better query complexity to $U_A$, eluding the lower bound of~\cite{lowOptimalQuantumSimulation2025} (this of course breaks down for large $t$, see Appendix~\ref{apx:conformal-maps-2} for a method to partially overcome this problem).

At the same time, by diagonalizing the incrementers $Q_n$ in Figure~1 of the main text with a quantum Fourier transform $F_n$, also the number of extra gates is better: we decompose $Q_n = F_n \Sigma F_n^\dag$, where
\begin{align*}
    \Sigma \ket{x} = \omega_n^{x} \ket{x}
\end{align*}
is the clock matrix of dimension $n$, which is easily implemented with $\bigO(n)$ controlled rotations. As the intermediate Fourier transforms cancel out, we only need to implement one $F_n$ at the beginning and one $F_n^\dag$ at the end. The count of all the extra gates, dominated by the complexity of the two QFTs, is thus $\bigO(n^2) = \bigO(\log^2 \frac{1}{\epsilon})$~\footnote{As stated in the main text, using an approximate QFT could be helpful in the general case, but in the particular case of the exponential we do not get any asymptotic improvement.}.

Note that, on the other hand, LCHS requires super-quadratic extra gates, mainly due to the implementation of the state-preparation oracle for the LCU~\cite[Theorem 4]{lowOptimalQuantumSimulation2025}.

\section{\gr{E. Simulating second-order wave dynamics}}

In this section we show how to simulate the evolution of a second-order differential equation of the form $\frac{d^2 u}{dt^2} = -Au$. The solution is given by
\begin{align*}
    u(t) = \cos(t \sqrt{A}) u(0) + A^{-1/2} \sin(t \sqrt{A})\dot{u}(0) \ .
\end{align*}
The task here is thus implementing the block-encodings $\cos(t \sqrt{A}), A^{-1/2} \sin(t \sqrt{A})$.

We remark that these dynamics are inherently unstable, unless $A$ has only real and positive eigenvalues (in which case a QSVT would be sufficient nonetheless).

We are thus going to assume that $t$ is small. In this situation, we only need to implement the two functions
\begin{align*}
    f_1(z) = \frac{1}{2 \cosh t} \cos(t \sqrt{z}), \ \ \ f_2(z) = \frac{z^{-1/2}}{2 \cosh t} \sin(t \sqrt{z}) \ ,
\end{align*}
which have subnormalization $2\cosh(t) = \Theta(1)$ for the short-time regime $t = \bigO(1)$ we are considering. Note that $f_1/f_2$ are even/odd in $\sqrt{z}$, and thus are single-valued analytic functions of $z$. Both these functions have an analytic extension to the whole complex plane, with a simple Taylor series converging uniformly everywhere:
\begin{align*}
    f_1(z) = \frac{1}{2 \cosh t} \sum_{n=0}^{\infty} \frac{(-1)^n t^{2n}}{(2n)!} z^n, \ \ \ f_2(z) = \frac{1}{2 \cosh t} \sum_{n=0}^{\infty} \frac{(-1)^n t^{2n+1} z^n}{(2n+1)!} \ .
\end{align*}

For fixed (small) $t$, the two Taylor series have convergence rates
\begin{align*}
    \bigO\qty(\frac{\log \frac{1}{\epsilon}}{\log \log \frac{1}{\epsilon}})
\end{align*}
which is also the number of copies of $A$ required to implement either function. Note that, on the other hand, previous methods in~\cite{wangQuantumSimulationNonHermitian2026a} for wave dynamics have a $\sqrt{1/\epsilon}$ dependence.

\section{F. Criteria of convergence for Taylor series (proof of Theorem 7)}

\begin{proof}
    We can write down the $k$-th Taylor coefficient using Cauchy's integral formula. Since $f(z)$ has an analytic extension to the disk $\abs{z} < R$, we can compute the integral along a circle $\gamma$ centered at $z = 0$ of radius $\rho$, for some $1 < \rho < R$.
    \begin{align}
        c_k = \frac{1}{2\pi i} \oint_{\gamma} \frac{f(z)}{z^k} \frac{dz}{z} = \frac{1}{2\pi} \int_0^{2\pi} \frac{f(\rho e^{i\theta})}{\rho^k e^{ik\theta}} d\theta
    \end{align}
    By a triangle inequality, we get that $\abs{c_k} \le \frac{M}{\rho}^k$. Using $\abs{z}\leq 1 < \rho$, the error estimate is thus bounded by
    \begin{align}
        \abs{f(z) - \sum_{k = 0}^N c_k z^k} & \le \sum_{k = {N+1}}^{\infty} \frac{M}{\rho^k} \le \frac{M}{\rho^N (\rho-1)}
    \end{align}
    and we obtain the claimed bound by letting $\rho \rightarrow R$.
\end{proof}

\section{\gr{G. Domain manipulation through conformal maps}}
\label{apx:conformal-maps-2}

The presented method exhibits an important limitation: the approximated function is required to be analytic and bounded by a constant throughout the whole unit disk, even if the numerical range of $A$ can be safely assumed to be in a portion of it. We can overcome this limitation by a manipulating the domain of the function via a conformal map from the numerical range of $A$ to the whole unit disk. As a proof of concept, we present here a specific application of this method to the exponential function, but it could be applied to any other function and from any region.

\begin{figure}
    \centering
    \begin{tikzpicture}[>=stealth, scale=2.5]
    % Define eta directly
    \def\eta{0.2}

    % Compute z_0 and rho using pgfmath directly from eta
    \pgfmathsetmacro{\zzero}{(1 - \eta*\eta) / (2 * \eta)}
    \pgfmathsetmacro{\rho}{(1 + \eta*\eta) / (2 * \eta)}

    % Bounding box to clip large objects (like the circle arc and line)
    \clip (-2.2, -2.2) rectangle (2.2, 2.2);

    % Grid and Axes
    \draw[very thin, gray!30] (-2.0, -1.5) grid (2.0, 1.5);
    \draw[->, thick] (-1.0, 0) -- (1.8, 0) node[right] {$\R$};
    \draw[->, thick] (0, -1.5) -- (0, 1.5) node[above] {$\I$};

    % Origin label
    \node[below left] at (0,0) {$0$};

    % Unit Disk (Red, no fill)
    \draw[red, thick] (0,0) circle (1);

    % Circle C_{\eta} drawn as an arc centered at z_0
    % Using arc (start_angle:end_angle:radius)
    \draw[blue, thick] (\zzero, 0) ++(145:\rho) arc (145:215:\rho);

    % Line Re z = -\eta (Dashed)
    %\draw[dashed, black!70, thick] (-\eta, -1.6) -- (-\eta, 1.6) node[above right] {$\operatorname{Re} z = -\eta$};

    % Key Points on the Axes
    % Point -\eta
    \filldraw[black] (-\eta, 0) circle (0.8pt) node[below left] {$-\eta$};

    % Annotations/Legend
    \node[red, below] at (0.6, 0.7) {$\D$};
    \node[blue, right] at (0.3, 1.2) {$C_{\eta}$};

\end{tikzpicture}
    \caption{Visualization of $C_{\eta}$ in the complex plane for $\eta = 0.2$. Note that as $\eta \rightarrow 0$, $C_\eta$ becomes the full right half-plane. Both the center $w_0$ and the radius $\rho$ scale as $\Theta(t)$ when $\eta = 1/t$.}
    \label{fig:cayley-2}
\end{figure}

We consider the implementation of $e^{-tA}$, where we assume to have stable dynamics ($A + A^\dag \succeq 0$). Naively applying $e^{-tz}$ is not possible, as $\sup_{z \in \D} \abs{f(z)} = e^t$. We could map the unit circle to the right half plane right before applying the exponential using a conformal map $\calM(z)$. Since $e^{-t \calM(z)}$ is analytic and bounded by $1$, it could be in principle approximated by GQSP. This $\calM$ would simply be the Cayley map
\begin{align*}
    \calM_{\eta}(z) = \rho z + w_0
\end{align*}
where we take $\rho = \frac{1 + \eta^2}{2\eta}, w_0 = \frac{1 - \eta^2}{2\eta}$. This transformation maps the unit circle $\D$ to the circle $C_{\eta}$ of center $w_0$ and radius $\rho$ (see Figure~\ref{fig:cayley-2}). The property of $C_{\eta}$ we will use is that $\Re w \ge -\eta$ for every $w$. The inverse map is the Cayley map:
\begin{align*}
    \calM^{-1}_{\eta}(w) = \frac{w - w_0}{\rho}
\end{align*}
mapping $C_{\eta}$ to $\D$. In the following we will set $\eta = \frac{1}{t}$, so that
\begin{align*}
    \calM_{\eta}(z) = \frac{(t^2 + 1) z + (t^2 - 1)}{2t}, \ \ \ \ \calM_{\eta}^{-1}(w) = \frac{2tw - (t^2 - 1)}{(t^2 + 1)}
\end{align*}
The idea is as follows: we first use a simple application of LCU to produce $\calM^{-1}_{\eta}(A)$ using only one copy of $A$ and then implement the function
\begin{align*}
    f(z) = e^{-t \calM_{\eta}(z)} = e^{-\frac{t^2 - 1}{2}} \cdot e^{-\frac{t^2 + 1}{2} z}
\end{align*}
Note that $\abs{f(z)} \le e$ on $\D$, so we can implement it by incurring only a constant subnormalization. Moreover, $f$ is an entire function whose Taylor series converges with rate $\bigO(t^2 + \log 1/\epsilon)$. This is of course not optimal~\cite{lowOptimalQuantumSimulation2025}, but notice that the dependence on the error is additive rather than multiplicative. Intuitively, the quadratic complexity comes from the fact that by applying $\calM_{\eta}^{-1}$ we are deliberately subnormalizing $A$ by a factor $1/t$, which has to be compensated by $\calM_{\eta}$.

We also remark that we used very simple forms of Cayley maps (namely mapping finite circles to finite circles), but it is possible to implement more complex conformal mappings expressible through generic rational functions using inverses by QSVT.

%\bibliography{refs}

\end{document}